\def\draft{1}  
    \documentclass[12pt]{article}
    \usepackage[letterpaper,hmargin=1in,vmargin=1in]{geometry}
    \usepackage{verbatim,sgame}

    \usepackage{natbib}

    \usepackage{tikz,bm}
    \usetikzlibrary{positioning,chains,fit,shapes,calc}
    \usepackage{float,subcaption}
    \floatstyle{boxed}
    \restylefloat{figure, graphicx}

    \usepackage{amsfonts,url,ifthen,epsfig,amsmath,amssymb,color,framed,hyperref,float,algorithm}
    \usepackage[noend]{algpseudocode}
    \usepackage{etoolbox}

    \makeatletter
    \let\original@footnotemark\footnotemark
    \newcommand{\align@footnotemark}{%
      \ifmeasuring@
        \chardef\@tempfn=\value{footnote}%
        \original@footnotemark
        \setcounter{footnote}{\@tempfn}%
      \else
        \iffirstchoice@
          \original@footnotemark
        \fi
      \fi}
    \pretocmd{\start@align}{\let\footnotemark\align@footnotemark}{}{}
    \makeatother

    \hypersetup{colorlinks=false}

    \ifnum\draft=1
    \newcommand{\Rnote}[1]{\begin{framed}\noindent \textcolor{red}{{#1}}\end{framed}} 
    \else
    \newcommand{\Rnote}[1]{}
    \fi

    \newcommand{\remove}[1]{}
    \newtheorem{theorem}{Theorem}
    \newtheorem{example}{Example}
    \newtheorem{assumption}{Assumption}

    \newtheorem{lemma}{Lemma}
    \newtheorem{corollary}{Corollary}
    \newtheorem{obs}{Observation}
    \newtheorem{definition}{Definition}
    \newtheorem{proposition}{Proposition}
    \newtheorem{remk}[theorem]{Remark}

    \def\FullBox{\hbox{\vrule width 8pt height 8pt depth 0pt}}

    \def\qed{\ifmmode\qquad\FullBox\else{\unskip\nobreak\hfil
    \penalty50\hskip1em\null\nobreak\hfil\FullBox
    \parfillskip=0pt\finalhyphendemerits=0\endgraf}\fi}

    \def\qedsketch{\ifmmode\Box\else{\unskip\nobreak\hfil
    \penalty50\hskip1em\null\nobreak\hfil$\Box$
    \parfillskip=0pt\finalhyphendemerits=0\endgraf}\fi}

    \newenvironment{proof}{\begin{trivlist} \item {\bf Proof:~~}}
      {\qed\end{trivlist}}
    
    \newenvironment{proofof}[1]{
      \medskip 

      {\textsc{Proof of #1.}~}}
      {\qed
    }

    \newcommand{\beq}{\begin{equation}}
    \newcommand{\eeq}{\end{equation}}
    \newcommand{\be}{\begin{enumerate}}
    \newcommand{\ee}{\end{enumerate}}
    \newcommand{\bi}{\begin{itemize}}
    \newcommand{\ei}{\end{itemize}}
    \newcommand{\bd}{\begin{description}}
    \newcommand{\ed}{\end{description}}
    \newcommand{\bc}{\begin{center}}
    \newcommand{\ec}{\end{center}}
    \newcommand{\bthm}{\begin{theorem}}
    \newcommand{\ethm}{\end{theorem}}
    \newcommand{\bdefi}{\begin{definition}}
    \newcommand{\edefi}{\end{definition}}
    \newcommand{\bcor}{\begin{corollary}}
    \newcommand{\ecor}{\end{corollary}}
    \newcommand{\blem}{\begin{lemma}}
    \newcommand{\elem}{\end{lemma}}
    \newcommand{\bexa}{\begin{example}}
    \newcommand{\eexa}{\end{example}}
    \newcommand{\bprop}{\begin{proposition}}
    \newcommand{\eprop}{\end{proposition}}

    \ifnum\draft=1
    \newcommand{\ronote}[1]{\begin{framed}\noindent \textcolor{red}{{Ronen's note: #1}}\end{framed}} 
    \newcommand{\ranote}[1]{\begin{framed}\noindent \textcolor{red}{{Rann's note: #1}}\end{framed}} 
    \else
    \newcommand{\ranote}[1]{}
    \newcommand{\ronote}[1]{}
    \fi

    \newcommand{\eqdef}{\mathbin{\stackrel{\rm def}{=}}}

    \def\real{\hbox{\rm\setbox1=\hbox{I}\copy1\kern-.45\wd1 R}}
    \def\neal{\hbox{\rm\setbox1=\hbox{I}\copy1\kern-.45\wd1 N}}

    \newcommand{\eps}{\varepsilon}

    \newcommand{\oa}{{\overline{a}}}
        \newcommand{\oo}{{\overline{\omega}}}
    
    \newcommand{\ang}[1]{\langle{#1}\rangle}
\newcommand{\ppi}{\bm{\pi}}

    \newcommand{\A}{{\mathcal A}}

    \newcommand{\R}{{\mathbb R}}
    \newcommand{\supp}{\mathrm{supp}}

\begin{document}

    \definecolor{myblue}{RGB}{80,80,160}
    \definecolor{mygreen}{RGB}{80,160,80}

\begin{titlepage}
    \title{Reaping the Informational Surplus in Bayesian Persuasion}
    

\author{Ronen Gradwohl\thanks{Department of Economics and Business Administration, Ariel University. Email: \texttt{roneng@ariel.ac.il}. Gradwohl gratefully acknowledges the support of National Science Foundation award number 1718670.} \and Niklas Hahn\thanks{Institute of Computer Science, Goethe University Frankfurt. Email: \texttt{nhahn@em.uni-frankfurt.de}. Hahn gratefully acknowledges the support of the German-Israeli Foundation grant I-1419-118.4/2017.} \and Martin Hoefer\thanks{Institute of Computer Science, Goethe University Frankfurt. Email: \texttt{mhoefer@em.uni-frankfurt.de}. Hoefer gratefully acknowledges the support of the German-Israeli Foundation grant I-1419-118.4/2017 and Deutsche Forschungsgemeinschaft grants DFG Ho 3831/5-1, 6-1, and 7-1.} \and Rann Smorodinsky\thanks{Faculty of Industrial Engineering and Management, The Technion -- Israel Institute of
    Technology. Email: \texttt{rann@ie.technion.ac.il}. Smorodinsky gratefully acknowledges United States-Israel Binational Science Foundation and National Science Foundation grant 2016734, the German-Israeli Foundation grant I-1419-118.4/2017, the Ministry of Science and Technology grant 19400214, Technion VPR grants, and the Bernard M. Gordon Center for Systems Engineering at the Technion.}}
    \date{}
\maketitle
	\begin{abstract}
The Bayesian persuasion model studies communication between an informed sender and a receiver with a payoff-relevant action, emphasizing
the ability of a sender to extract maximal surplus from his informational advantage. In this paper
we study a setting with multiple senders, but in which the receiver interacts with only one sender of his choice: senders commit to signals and the receiver then chooses, at the {\em interim} stage, with which sender to interact. 
Our main result is that whenever senders are even slightly uncertain about each other's preferences, the {\em receiver} receives all the informational surplus in all equilibria of this game.

	
	\end{abstract}

    \thispagestyle{empty}

\end{titlepage}

    \renewcommand{\thefootnote}{\arabic{footnote}}
    \setcounter{footnote}{0}


\section{Introduction}

The celebrated model of information design known as Bayesian persuasion \citep{KamenicaG11}  studies a setting where some agents, typically known as {\em senders}, receive private information. Other agents, known as {\em receivers}, must choose an action that affects their own payoff as well as the senders'. The senders' challenge, and the focus of the literature, is to decide how much of the information they should share with the receivers.  In particular, the Bayesian persuasion model captures settings where the informed senders can commit to a {\em signal}---a distribution over messages that depends on the
senders' information---prior to learning their private information.

Although senders take no action, the combination of access to private information with the ability to commit to a signal proves to be quite advantageous. However, when multiple senders are involved, this advantage may decline. Indeed, recent work discusses the deterioration in senders' payoffs as more senders compete, 
at the same time benefiting the single receiver. We discuss some of these results in Section~\ref{subsec:related}.

In this paper we provide a new model of competition among senders. In our model there are multiple senders and a state space for each of the senders and receiver. States are a priori unknown, but there may be almost arbitrary correlation between different players' states.
As always, senders commit to a signal prior to receiving any private information. In contrast with other models, however,
our receiver is restricted to receiving one message, from a sender of his choice.  

%


From a descriptive perspective, there are many settings in which the receiver may lack the attention to receive input from all senders.
For example, the receiver may be
a policymaker contemplating the implementation of some policy, and the senders may be experts who are knowledgeable about the policy's projected implications.
Consulting with all the experts may be costly; instead, the policymaker might just choose one with whom to consult.\footnote{During the recent COVID-19 pandemic many academics tried to provide input to their own governments, but they soon realized there is congestion in advice, and could not make themselves heard.}
Alternatively, the receiver may be a judge or court of law and the senders prosecutors arguing a case \citep[as in the canonical example of][]{KamenicaG11}. 
The court may be limited in the number of cases it can hear, and so may have to choose with which prosecutors to interact. 

From a normative perspective,  when the receiver has flexibility in choosing how much advice to receive, he may prefer to limit the amount of advice.
Indeed, our main result is that in {\em all} equilibria of the game, the receiver learns {\em all} of his payoff-relevant information. That is, with as few as two competing senders, the receiver reaps all of the informational surplus. We obtain this strong dichotomy between the single- and multiple-sender cases in a very general model. We allow for an arbitrary information structure and arbitrary utility functions of the senders, as long as senders are not perfectly aligned in their objectives. 
Without this assumption---that is, if senders are perfectly aligned---then the multiple-sender case is nearly identical to the the single-sender case, 
and so the informational surplus is primarily enjoyed by the senders.

Our main result leads to two counterintuitive observations: First, from the receiver's point of view, the restriction to interacting with just a single sender is actually a benefit, and in fact should be self-imposed. 
Second, from the senders' point of view, commitment power could be a double-edged sword. A single sender interacting with a receiver is always better off 
with commitment power, but with more than one sender that same commitment power could be strictly harmful. Here, senders may be better off forgoing this power.

\subsection{Motivating Examples}
We now discuss these observations in greater detail.

\begin{example}[Inspired by an example of \citeauthor{KamenicaG11}, \citeyear{KamenicaG11}.] \rm
\label{ex:expost}
There is a government policy that may be beneficial or harmful and a policymaker who can either implement the policy (action $P$) or 
maintain the status quo (action $Q$). 
Additionally, there are two experts, and each is 
either biased or unbiased. The biased type always wants the policy implemented \citep[similarly to the prosecutor in the example of][]{KamenicaG11} whereas the unbiased type is aligned with the policymaker.
Each player obtains utility 1 when his preferred action is taken, and utility 0 otherwise.

The prior distribution over the information is as follows: The policy is beneficial with probability $0.5$ and each of the experts is unbiased with probability 
$\eps > 0$, all independently. Each expert learns his own type and whether or not the policy is beneficial.

Before we discuss competition between experts we consider the optimal signal for a single expert, absent competition. As a single, unbiased expert is fully aligned with the policymaker he will fully disclose the policy's benefit. On the other hand, a biased expert will always recommend implementation. As a result, the policymaker will 
implement the policy upon hearing a recommendation to implement, and will maintain the status quo otherwise. This yields an expected utility of 1 to the expert and 
$0.5+\eps/2$ to the policymaker.

We now return to the two-sender setting and consider two variants of the game. In the first, the policymaker chooses one expert with whom to interact (which is the
model discussed in the paper).
 In the second, the policymaker interacts with both.
\begin{itemize}
\item \textbf{Two senders and one signal:} This is the setting we will analyze in the rest of the paper. In this setting the policymaker chooses a single expert,
and observes the realization only of that expert's signal. We apply our main result, Theorem~\ref{thm:mixed}, to deduce that in this setting the policymaker will always take his correct action. This means that the experts' utility drops to $0.5+\eps/2$, while the policymaker's increases to 1. Note that this conclusion holds for any $\eps>0$. However, at $\eps=0$ we witness a striking discontinuity, as all experts can now choose the sender-optimal signal from the single-expert case, and will then fully enjoy their  
informational advantage.\footnote{That is, each expert always recommends implementation, and the policymaker obliges as he is indifferent. A slight
perturbation of the prior upwards, or of the expert's signals, renders this his strictly optimal action.}

\item\textbf{Two senders and two signals:} What if the policymaker observes both experts' signal realizations? Consider the case where all experts adopt the optimal course of action for the single-expert setting. As the policymaker observes all realizations of signals his optimal strategy is to maintain the status quo if one or more 
of the experts
recommends this, and otherwise implement the policy. For small $\eps$, the result is that the experts  once again (almost) fully extract the informational surplus, while the policymaker expects a utility close to $0.5$: With probability $(1-\eps)^2$, both experts are biased and recommend implementation, and the policymaker
obliges, leading to an expected utility of $0.5$ for the latter in this case. 

To see why this strategy profile constitutes an equilibrium in the game note that an unbiased expert always gets utility of $1$, and so cannot improve. A 
biased expert gets utility 1 when the other expert is also biased, since then both recommend their optimal action, implementation. 
The only situation in which a biased expert does not get his maximal utility is when the other expert is unbiased and the policy harmful. However, in this case 
he cannot change the policymaker's action from $Q$ to $P$, since the policymaker knows that the other expert's recommendation is aligned with his own 
preference.
\end{itemize}
\end{example}
Thus, in the example above the receiver strictly prefers to choose and interact with just one sender rather than to observe both senders' signal realizations.\\

A sender that can commit to his future course of action enjoys an obvious advantage over one who cannot, whenever he monopolizes information. To see this, observe that a sender with commitment power can always simulate one that has no such power. In the competitive setting, however, all the informational surplus is enjoyed by the receiver, and so commitment power is no longer an obvious benefit. Indeed, the following example demonstrates the possibility that senders engaged in competition may prefer to have no commitment power:

\begin{example}\label{ex:one_lobbyist}\rm
Consider a regulator who is contemplating the imposition of restrictions in the market for electronic cigarettes. The regulator can tap in to one of two leading experts
(lobbyists) on this market. Expert 1 supports the restrictions only if e-cigarettes increase lung cancer rates, modeled as a binary state: healthy ($H$) and unhealthy ($U$). Expert 2 supports the restrictions only when traditional cigarettes are more popular with women ($W$) than with men ($M$). Finally, the regulator is interested in imposing restrictions only if electronic cigarettes increase smoking rates among youth ($Y$ vs.\ $O$). Each of the players receives a utility of 1 if restrictions are imposed when they prefer them, and 0 otherwise.

There are 8 states of the world, and suppose the prior distribution over them is given by 
$$
\begin{array}{llll}
P(M,H,Y)= 0.18 & P(M,H,O) = 0.08 & P(M,U,Y)=0.12 & P(M,U,O)=0.12\\
P(W,H,Y) = 0.12 & P(W,H,O)=0.12 & P(W,U,Y)=0.08 & P(W,U,O)=0.18
\end{array}
$$
Each of the experts learns his own payoff-relevant state as well as the receiver's payoff-relevant state, but not that of the other expert. 
Our main result, Theorem~\ref{thm:mixed}, implies that
 whenever experts commit to signals and only one sender is chosen by the regulator, the regulator always takes his preferred action. Thus, 
 the unique equilibrium payoffs are $1$ to the regulator and $0.4$ to each expert.

In contrast, absent any commitment power the two experts could provide no information in equilibrium,\footnote{This follows from standard
cheap-talk arguments.} the regulator would take an arbitrary action as he is indifferent, and experts would enjoy a higher payoff of $0.5$.%
\footnote{The interested reader can verify that with a single sender the commitment power is advantageous. Assume only one expert, say the first one, and note that
 the induced prior is $P(H,Y)= 0.3, P(H,O)=0.2, P(U,Y)=0.2, P(U,O)=0.3$. With commitment power the expert expects a payoff of $0.9$, above and beyond the payoff of $0.5$ he expects absent such power.}
\end{example}

Thus, this example demonstrates that commitment power could be a double-edged sword for the senders: It is beneficial when there is a single sender, but
can be strictly harmful when there are multiple senders.

\subsection{Related Literature}
\label{subsec:related}
The
study of Bayesian persuasion was initiated by \citet{AumannM66} and came back into focus more recently
following the work of \citet{KamenicaG11}, leading to a
plethora of variants \citep[e.g.,][]{CelliCG20,ArieliB16,Ely17,ElyFK15,Au15,EmekFGLT12,AlonsoC16,Kolotilin15,GoldsteinL18,RabinovichJJX15,koessler2019long}. \citet{Kamenica19} provides a good overview of Bayesian persuasion and some of the extensions. 

Our research contributes to work on competing senders and the effect this competition has on the amount of information revealed to the receiver. This work includes
\citet{GentzkowK17Bayesian}, who study a model in which the decision maker receives the realized signals from every sender before making a decision. They show that increasing the number of senders cannot decrease the amount of information revealed to the receiver. Their work is supplemented by that of \citet{LiN18}, who show that the assumptions in the previous reference are critical and the result does not hold if any are violated.
We show a stark contrast between the model of \cite{GentzkowK17Bayesian} and one where the receiver 
is restricted to receiving the realized signal from a single sender. In the latter model the signals of competing senders in equilibrium are fully informative to the receiver, whereas in the former there are cases where equilibria are not fully informative.

\citet{AuK20} study the case of $n$ senders with a valuation drawn independently from a single known distribution competing for the patronage of a receiver. Only a single sender can be chosen by the receiver and only this sender gets a positive payoff. The receiver's payoff is determined by the chosen sender's valuation. The paper shows that an increase in the number of senders increases the amount of information revealed by each individual sender. In contrast to this, our model allows for almost arbitrary correlation between the senders and the receiver. Additionally, we allow multiple senders to profit from the receiver's decision. Our setting shows that already with two senders, senders are fully informative.

 
Finally, \cite{gentzkow2017competition}, compare the information revealed when senders compete to that of collusive senders in the same scenario. They identify a condition on the information structure which is necessary and sufficient for competition to increase the informativeness of signals in a broad class of settings. We note that this condition does not hold in our setup. They also show that their result does not hold if mixed equilibria are permitted, whereas our result holds for both
pure and mixed equilibria. We use a different approach and focus on commitment power, comparing the information revealed and payoff for a single sender to that of multiple competing senders.

\section{Model}\label{sec:model}
Some of the notation is adapted from \cite{gentzkow2017competition}.
There are $n$ senders, denoted $\{1,\ldots,n\}$, and a single receiver, denoted $R$. The state space is finite and is of the form $\Omega=\Omega_1\times\ldots\times\Omega_n\times\Omega_R$, with a typical element $\omega=(\omega_1,\ldots,\omega_n,\omega_R)$. 
The senders and receiver share a common prior $\mu_0$ over $\Omega$. 

The receiver is the only player with a payoff relevant action. Let $A$ denote this set of actions and
$u_R:\A\times \Omega_R  \mapsto \R$ and $u_i:\A\times \Omega_i  \mapsto \R$ the payoff functions of the receiver and sender $i$, respectively. 
Note that each player's payoff depends only on the corresponding entry in the state space.
This generalizes the standard model in which there is a common, payoff-relevant state $\Omega$ to all players, as the prior could
be such that $\omega_R=\omega_1=\ldots=\omega_n$ always. Our assumptions on utilities will limit this generality to some extent, however.

We make two assumptions. The first, made solely to simplify the exposition, states that each player has a unique optimal action in each state:

\begin{assumption}\label{assumption:unique-optimal}
	For every $j\in\{1,\ldots,n,R\}$ and $\omega_j\in\Omega_j$, the set $\arg\max_{a\in \A} u_j(a,\omega_j)$ is a singleton.
\end{assumption}

The second assumption is the main assumption underlying our theorem. In fact, without such an assumption the result will not hold,
as demonstrated by Example~\ref{ex:expost} with $\eps=0$. In words, given the receiver's and any sender's states,
there is still some residual uncertainty about other senders' states, and in particular on the possibility that their preferences are aligned with the receiver.
\begin{assumption}\label{assumption:possibly-aligned}
	For every pair of senders $i\neq j$ and every $(\omega_{j},\omega_R)\in \Omega_j\times\Omega_R$ that has positive probability under $\mu_0$, 
	there exists an $\omega_i\in\Omega_i$ with
	$P(\omega_i \mid (\omega_{j},\omega_R))>0$
	for which $\arg\max_{a\in \A} u_i(a,\omega_i) = \arg\max_{a\in \A} u_R(a,\omega_R)$.
\end{assumption}

This assumption warrants some discussion, as it is the main driving force behind our result. First, note that the assumption is ``weak" in that it is implied
by the following ``standard" assumptions: that the prior $\mu_0$ has full support, and that there are no undesirable actions---for 
each sender $i$ and each action $a$, there is some state $\omega_i$
for which $a= \arg\max_{a'\in \A} u_i(a',\omega_i)$. Second, note that 
Assumption~\ref{assumption:possibly-aligned} will be satisfied for any prior when senders' preferences are perturbed a bit so that each sender,
 independently but with arbitrarily small probability, is aligned with the receiver (as in Example~\ref{ex:expost}).


Next, in order to describe the game we require some preliminary notation. A signal for player $i$ is composed of an abstract message space, $M_i$, and a function $\pi_i: \Omega_{iR} \mapsto \Delta(M_i)$, where $\Omega_{iR} =\Omega_{i} \times\Omega_{R}$ and  $\Delta(\cdot)$ denotes the corresponding set of probability distributions. 
Let $\Pi_i$ denote the set of all signals of sender $i$.

The persuasion game we study proceeds as follows. 
First, each sender $i$ chooses a distribution $\ppi_i$ over $\Pi_i$. The receiver observes the vector $\left(\pi_1,\ldots,\pi_n\right)$ drawn from 
$\left(\ppi_1,\ldots,\ppi_n\right)$, after which he chooses one of the players, say $j$. The state is then realized, sender $j$ learns the state 
$\omega_{jR}\eqdef (\omega_j,\omega_R)$ and the receiver gets (only) the message $\pi_j(\omega_{jR})$ sent by the chosen player $j$. 
Finally, $R$ takes an action in $\A$ and payoffs of all players are realized.\footnote{We stress that the receiver chooses a sender with whom to interact {\em after} he observes the signal $\pi_i$ drawn from each sender $i$'s
distribution $\ppi_i$. Were he to make this choice after observing only the distributions $\ppi_i$, the model would be identical to one restricted to pure
strategies, since each $\ppi_i$ is equivalent to some $\pi_i$ in terms of the senders' and receiver's eventual utilities.}

%




Without loss of generality we assume
that $M_i \subseteq \A$, and interpret the message realized by the signal as a recommended action to the receiver. A signal $\pi_i$
is {\em incentive compatible (IC)} if, upon any realization $a\in \A$ of $\pi_i$, the receiver's optimal action is $a$. We invoke the revelation principle and assume, without loss of generality, that senders are restricted to IC signals. Hereinafter we let $\Pi_i$ denote the set of all such IC signals.

Denote by $v_R(\pi)$ the expected utility of the receiver when he takes his optimal action following every realization of the signal $\pi$.
Denote by $v_i(\pi)$ the expected utility of sender $i$ when the receiver takes these actions. We write
$\pi\succeq_j\pi'$ when $v_j(\pi)\geq v_j(\pi')$, and $\pi\succ_j\pi'$ when $v_j(\pi)> v_j(\pi')$, for any $j\in\{1,\ldots,n,R\}$.

Note that for any sender $j$, a signal $\pi_j$ and realization $a\in \A$ of $\pi_j$ induce a posterior belief over $\Omega$. Denote this belief by $\ang{\pi_j  \mid a}$, and
by $\ang{\pi_j}$ the distribution of beliefs induced by realizations of $\pi_j$. Furthermore, denote by $\ang{\pi_j}_i$ the marginal
distribution of $\ang{\pi_j}$ over $\Omega_i$, and by $\ang{\pi_j}_{(i,R)}$ the marginal distribution of $\ang{\pi_j}$ over $\Omega_i\times\Omega_R$.


Finally, we will be interested in the Nash equilibria of this game. We assume that when the receiver is indifferent between multiple senders, he chooses one
of them uniformly
at random---this is akin to the standard assumption in Bayesian persuasion that the receiver breaks ties in favor of the sender.\footnote{Our result
holds also for weaker assumptions on the tie-breaking rule---see Appendix~\ref{sec:decision-rule-extension}.} To this end, 
fix the decision rule $D:\Pi_1\times\ldots\times\Pi_n\mapsto\Delta\left(\{1,\ldots,n\}\right)$ of the receiver as a uniformly-random choice of 
a sender from the set $\left\{j:\pi_j\succeq_R \pi_i ~\forall i\in\{1,\ldots,n\}\right\}$.

\begin{definition}
A {\em Nash equilibrium (NE)} is a profile  $\ppi=\left(\ppi_1,\ldots,\ppi_n\right)$ 
with the following property: There do not exist a sender $i$ and distribution $\ppi_i'$ over signals with
$$E\left[v_i(\ppi'_{D(\ppi')})\right] > E\left[v_i(\ppi_{D(\ppi)})\right],$$
where $\ppi'$ is the profile $\ppi$ but with $\ppi_i'$ replacing $\ppi_i$.
\end{definition}


\section{Fully-Informative Equilibrium}\label{sec:equilibrium}
We are interested in fully-informative signals, in which the receiver obtains enough information to always take his optimal action. Formally,
\begin{definition}
	A signal $\pi_i$ is {\em fully informative} if for every $a\in\supp(\pi_i)$ and every $\omega_R\in\supp\left(\ang{\pi_i \mid a}_R\right)$ the recommended
	action $a=\arg\max_{b\in A} u_R(b,\omega_R)$. 
\end{definition}
Moreover, we are interested in fully-informative profiles of signals:
\begin{definition}	
	A profile $\ppi$ is {\em fully informative} if, with probability 1, the profile $\left(\pi_1,\ldots,\pi_n\right)$ drawn from $\ppi$ satisfies the following:
	For every sender $j$  that is chosen by the receiver with positive probability it holds that $\pi_j$ is fully informative.
\end{definition}

Our main result is then:
\begin{theorem}\label{thm:mixed}
	Every NE of the game is fully informative.
\end{theorem}



The main idea underlying the proof of Theorem~\ref{thm:mixed} is best illustrated by the simpler case in which $\ppi$ is a pure profile,
 say $\left(\pi_1,\ldots,\pi_n\right)$. Suppose that in this profile, player $j$ is always
chosen by the receiver, but that $\pi_j$ is not fully informative. In this case we can construct a profitable deviation for any other player $i$.
First, in Lemma~\ref{lem:simulate} we show that this other player can {\em simulate} the signal of player $j$, 
by playing a signal $\pi_i'$ that yields the same marginal distribution over
his own and the receiver's states as $\pi_j$. Because of this equivalence in marginal distributions, both the receiver and player $i$ are indifferent between 
$\pi_j$ and $\pi'_i$. Then, in Lemma~\ref{lem:i-improve} we show that player $i$ can modify $\pi_i'$ to yield a signal $\pi_i''$ that both he and the receiver
strictly prefer to $\pi_i'$ and hence to $\pi_j$. The way he does this is by making $\pi_i'$ a touch more informative (to improve the receiver's utility),
but only when his own preferences are aligned with the receiver's (to also improve his own utility). This is feasible for $i$, because he knows his own state
while $j$ does not.
This deviation thus yields a profile in which the receiver always chooses $\pi_i''$, which player $i$ strictly prefers to the current profile in which $\pi_j$ is chosen.

\subsection{Preliminary Lemmas}
We begin with some lemmas that will be used in the proof of Theorem~\ref{thm:mixed}.

\begin{lemma}\label{lem:simulate}
	For every pair of senders $i\neq j$ and every signal $\pi_j\in\Pi_j$ there exists a signal $\pi_i\in\Pi_i$ for which $\ang{\pi_i}_{(i,R)}=\ang{\pi_j}_{(i,R)}$.
\end{lemma}

In words, sender $i$ can simulate any signal of $j$ so that the posterior distribution over the payoff relevant states of $i$ and $R$ equals that induced by $j$'s signal.

\begin{proof}
	The proof is by construction. Fix a pair of senders $i\neq j$ and a signal $\pi_j\in\Pi_j$. For each action $a\in\supp(\pi_j)$, recall that $\pi_j$ induces
	a distribution $\ang{\pi_j \mid a}_{(i,R)}$ over states $\Omega_i\times\Omega_R$, namely a probability $P(\omega_{iR} \mid \pi_j=a)$ for each state
	$\omega_{iR}$.
	Let the signal $\pi_i$ be the distribution over $A$ generated as follows: In each state $\omega_{iR}$, signal $\pi_i$ takes each value $a\in \A$ with probability
	$P(\pi_j=a \mid \omega_{iR})$, where
	$$P(\pi_j=a \mid \omega_{iR}) = \frac{P(\omega_{iR} \mid \pi_j=a) \cdot P(\pi_j=a)}{P(\omega_{iR})}.$$
	In this construction,
	\begin{align*}
	P(\omega_{iR} \mid \pi_i=a)&=\frac{P(\pi_i=a \mid \omega_{iR})\cdot P(\omega_{iR})}{\sum_{\omega_{iR}'} P(\pi_i=a \mid \omega_{iR}')\cdot P(\omega_{iR}')}\\
	&=\frac{P(\pi_j=a \mid \omega_{iR})\cdot P(\omega_{iR})}{\sum_{\omega_{iR}'} P(\pi_j=a \mid \omega_{iR}')\cdot P(\omega_{iR}')}\\
	&=P(\omega_{iR} \mid \pi_j=a)
	\end{align*}
	for every $\omega_{iR}$ and $a$. Thus, $\ang{\pi_i}_{(i,R)}=\ang{\pi_j}_{(i,R)}$.
\end{proof}

Note that in Lemma~\ref{lem:simulate}, the marginal distributions over $\omega_R$, given any recommended action $a$, is equal for $\pi_j$ and $\pi_i$.
Therefore, if $\pi_j$ is IC, then so is $\pi_i$.

	
\begin{lemma}\label{lem:assumption}
	For every pair of senders $i\neq j$, signal $\pi_j\in\Pi_j$, realization $a\in\supp(\pi_j)$, and $\omega_R\in\supp(\ang{\pi_j\mid a}_R)$
	 there exists $\omega_i\in\supp(\ang{\pi_j\mid a}_i)$ for which
	$\arg\max_{b\in A} u_i(b,\omega_i)=\arg\max_{b\in A} u_R(b,\omega_R)$.
\end{lemma}

\begin{proof} Fix some realization $a\in\supp(\pi_j)$ and state $\omega_{jR}\in\supp(\ang{\pi_j\mid a}_{(j,R)})$.
	By Assumption~\ref{assumption:possibly-aligned},
	for every pair of senders $i\neq j$ and every $\omega_{jR}\in \Omega_j\times\Omega_R$, there exists an $\omega_i\in\Omega_i$ with
	$P(\omega_i \mid \omega_{jR})>0$
	for which $\arg\max_{a\in \A} u_i(a,\omega_i) = \arg\max_{a\in \A} u_R(a,\omega_R)$. It remains to show that this $\omega_i\in\supp(\ang{\pi_j\mid a}_i)$:
	\begin{align*}
	P(\omega_i \mid \pi_j=a) &= \frac{P(\omega_i)\cdot P(\pi_j=a \mid \omega_i) }{P(\pi_j=a)}\\
	&=\frac{P(\omega_i)\cdot\sum_{\omega_{jR}'} P(\pi_j=a \mid \omega_i,\omega_{jR}')\cdot P(\omega_{jR}' \mid \omega_i)}{P(\pi_j=a)}\\
	&=\frac{P(\omega_i)\cdot\sum_{\omega_{jR}'} P(\pi_j=a \mid \omega_{jR}')\cdot \frac{P(\omega_i \mid \omega_{jR}')\cdot P(\omega_{jR}')}{P(\omega_i)}}{P(\pi_j=a)}\\
	&\geq\frac{P(\omega_i)\cdot P(\pi_j=a \mid \omega_{jR})\cdot \frac{P(\omega_i \mid \omega_{jR})\cdot P(\omega_{jR})}{P(\omega_i)}}{P(\pi_j=a)}\\
	&>0,
	\end{align*}
	where the strict inequality follows from Assumption~\ref{assumption:possibly-aligned} and the fact that all other terms are also
	strictly positive.
	\end{proof}

%

\begin{lemma}\label{lem:i-improve}
	Fix a pair of senders $i\neq j$ and a signal $\pi_j\in\Pi_j$ that is {\em not} fully informative.
	Then there exists a signal $\pi_i\in\Pi_i$ for which both $\pi_i\succ_R\pi_j$ and $\pi_i\succ_i\pi_j$.
\end{lemma}

\begin{proof}
	Since $\pi_j$ is not fully informative, there exist an action $\oa\in\supp(\pi_j)$ and a state $\oo\in\supp\left(\ang{\pi_j \mid \oa}_R\right)$
	such that $\oa\neq \arg\max_{b\in A} u_R(b,\oo)$.

	
	By Lemma~\ref{lem:simulate} there exists a signal
	$\pi_i'$ for which $\ang{\pi_i'}_{(i,R)}=\ang{\pi_j}_{(i,R)}$. This implies that both $v_i(\pi_i')=v_i(\pi_j)$ and $v_R(\pi_i')=v_R(\pi_j)$.

	Consider the set $\Omega_R^{\overline{a}}\eqdef\{\omega_R\in\supp\left(\ang{\pi_i'\mid \overline{a}}_R\right)\}$. For every $\omega_R\in\Omega_R^{\overline{a}}$
	let $b^{\omega_R} = \arg\max_{b\in A}u_R(b,\omega_R)$ be the optimal action for the receiver in state $\omega_R$.
	By Lemma~\ref{lem:assumption},
	for every $\omega_R\in\Omega_R^{\overline{a}}$ there exists $\omega_i^{b^{\omega_R}}\in\supp(\ang{\pi_i'\mid \overline{a}}_i)$ for which
	$b^{\omega_R}=\arg\max_{b\in A} u_i(\omega_i^{b^{\omega_R}},b)$.
	
	We now construct the signal $\pi_i$, as follows:
	\begin{enumerate}
		\item Given any state, generate the recommendation $a'$ according to $\pi_i'$.
		\item If the realized recommendation is $a' \neq \overline{a}$, signal $\pi_i$ is realized as $a'$.
		\item If the realized recommendation is $a'=\overline{a}$ and the state is $\left(\omega_i^{b^{\omega_R}},\omega_R\right)$
		for some $\omega_R\in\Omega_R^{\overline{a}}$, then:
		\begin{enumerate}
			\item with probability $\eps^{\omega_R}$ (to be determined below) signal $\pi_i$ is realized as
			$b^{\omega_R}$, and
			\item with probability $1-\eps^{\omega_R}$ the signal is realized as $a'$.
		\end{enumerate}
		\item Otherwise (the realized recommendation is $a'=\overline{a}$ but the state is {\em not} $\left(\omega_i^{b^{\omega_R}},\omega_R\right)$), 
		signal $\pi_i$ is realized as $a'$.
	\end{enumerate}
	
	We will set $\eps^{\omega_R}$ in 3(a) above so that the following is achieved: if the realized recommendation of $\pi_i'$ is $\overline{a}$, then with
	some small probability, $\pi_i$ is fully informative, and is realized as $b^{\omega_R}=\arg\max_{a\in\A} u_R(a,\omega_R)$. 
	Furthermore, in the case when $\pi_i$ is fully informative, the state $\omega_R$ is revealed (through the
	recommended action $b^{\omega_R}$) only when the state of player $i$ is $\omega_i^{b^{\omega_R}}$. That is, the signal reveals the optimal action
	for the receiver precisely in those cases when that same action is also optimal for sender $i$.
	
	To achieve this, set $\Lambda^a=\{\omega_R:a=\arg\max_{b}u_R(b,\omega_R)\}$.
	We will set $\eps^{\omega_R}$ so that, conditional on realizing 3(a) above, each recommendation $a$ is realized
	with probability proportional to $P(\omega_R\in\Lambda^a  \mid \pi_i'=\overline{a})$. Thus, this is akin to playing a fully informative signal with some
	small probability.
	Formally, set
	$$\eps^{\omega_R}=\frac{\eps\cdot P(\omega_R \mid \pi_i'=\overline{a})}{P\left(\left(\omega_i^{b^{\omega_R}},\omega_R\right) \mid \pi_i'=\overline{a}\right)}$$
	for $\eps>0$. Set $\eps$ small enough so that $\eps^{\omega_R}\leq 1$ for all $\omega_R\in\Omega_R^{\overline{a}}$.
	
	Hence, if the realized recommendation of $\pi_i'$ is $a'=\oa$ and the state is some $\omega_R\in\Omega_R^{\oa}$,
	then the signal $\pi_i$ is realized as $a$ with the following probability. For $a\neq \oa$,
	\begin{align*}
	P\left(\pi_i=a \mid \pi_i'=\oa\right) &= \sum_{\omega_R\in \Lambda^a\cap\Omega_R^{\oa}} \eps^{\omega_R}\cdot P\left( \omega_i^{a},\omega_R \mid \pi_i'=\oa\right)\\
	&=\sum_{\omega_R\in \Lambda^a\cap\Omega_R^{\oa}} \eps\cdot P(\omega_R \mid \pi_i'=\oa)\\
	&=\eps\cdot P\left(\omega_R\in\Lambda^a  \mid \pi_i'=\oa\right).
	\end{align*}
Because the signal $\pi_i$ thus constructed is sometimes realized as $\pi_i'$ and sometimes is fully informative, $\pi_i\succ_R \pi_i'$. Furthermore,
since the fully informative recommendation is realized precisely when sender $i$ and the receiver's preferences are aligned, $\pi_i\succ_i \pi_i'$.
Combining this with the facts above that $v_i(\pi_i')=v_i(\pi_j)$ and $v_R(\pi_i')=v_R(\pi_j)$ yields the claimed result.
\end{proof}

Finally, the following lemma essentially extends Lemma~\ref{lem:i-improve} to mixed strategies by the senders.
It states that if in some profile a sender plays signals that are almost never chosen, then he has a profitable deviation. 
\begin{lemma}\label{lem:not-chosen}
Fix a NE $\ppi$ that is not fully-informative. Then no sender $i$'s strategy $\ppi_i$ has positive measure on any set of signals that is chosen by the receiver
with probability 0.
\end{lemma}

Before proving the lemma we develop some notation. For a profile $\ppi$, let $T_i = \supp\left(v_R(\ppi_i)\right)$ be the set of possible 
receiver payoffs attainable by a signal in sender $i$'s support, and denote by $\tau_i=\min T_i$.\footnote{Recall that the support of a random variable
is the smallest {\em closed} subset of values with measure 1, and so the minimum of $T_i$ exists.} 
Furthermore, denote by $\tau=\max_i \tau_i$. Observe that the receiver will never choose a signal $\pi_i$ with $v_R(\pi_i)<\tau$, since there will
always be some signal leading to higher utility available. Also, note that all signals $\pi_i$ with $v_R(\pi_i)>\tau$ are potentially
chosen in equilibrium.

Finally, we denote by $\ppi_{-i}$  the profile $\ppi$ with sender $i$ removed. We also 
slightly misuse notation and write $D(\pi_{-i})$ as a uniformly-random choice of 
a sender from the set $\left\{j:\pi_j\succeq_R \pi_k ~\forall k\in\{1,\ldots,n\}\setminus\{i\}\right\}$.

\begin{proofof}{Lemma~\ref{lem:not-chosen}}
$\ppi_i$ is $i$'s equilibrium strategy, which implies that $i$ is indifferent between all but a measure 0 of the signals in the support of $\ppi_i$.
That is, for each such signal, sender $i$ considers the expected utility from that signal: the probability that it is chosen by the receiver times the utility
 $i$ derives from it.
In equilibrium, the expected utility of a particular signal is equal for all signals. 

Now, suppose towards a contradiction that $\ppi_i$ has positive measure on a set of signals that is chosen by the receiver with probability 0.
Thus, by the above, $i$ sometimes plays signals that are almost never chosen by the receiver. This leads to utility arbitrarily close to
$E\left[v_i(\pi_{D(\ppi_{-i})})\right]$ for sender $i$. The equilibrium indifference then implies that this same expected utility is obtained by sender $i$
for all the signals in the support of $\ppi_i$. Thus, his equilibrium payoff is arbitrarily close to $E\left[v_i(\pi_{D(\ppi_{-i})})\right]$.

We now construct a  profitable deviation for sender $i$, as follows. 
Let $j\neq i$ and $\pi_j$ satisfy $v_i(\pi_j)\geq v_i(\pi_k)$ for all senders $k\neq i$ and signals $\pi_k\in\supp(\ppi_k)$ with $v_R(\pi_k)\geq \tau$,
and for which $v_R(\pi_j)\geq\tau$.
That is, $\pi_j$ is the best signal for sender $i$ played in equilibrium by any other sender $j$ that is potentially chosen by the receiver.
If $\pi_j$ is not fully-informative, then
by Lemma~\ref{lem:i-improve} there exists a signal $\pi_i'$ for sender $i$ that is strictly better than $\pi_j$ for both $i$ and the receiver. Sender
$i$'s profitable deviation is then to play $\ppi_i'$, which is the same as $\ppi_i$ except that it replaces
the set of signals that is chosen by the receiver with probability 0 with $\pi_i'$.
This deviation is profitable, contradicting the assumption that $\ppi$ is a NE.

If $\pi_j$ is fully informative, then there are two cases. First, if $v_i(\pi_j)>E\left[v_i(\pi_{D(\ppi_{-i})})\right]$, then sender $i$ has a profitable deviation to always play
a fully-informative signal. If $v_i(\pi_j)=E\left[v_i(\pi_{D(\ppi_{-i})})\right]$, then it must be the case that for all $k\neq i$ and signals $\pi_k\in\supp(\ppi_k)$
for which $v_R(\pi_k)\geq \tau$ the equality
 $v_i(\pi_k)=E\left[v_i(\pi_{D(\ppi_{-i})})\right]$ holds. If some such $\pi_k$ is not fully-informative, then sender $i$ can invoke the same profitable deviation as
 above, by simulating and improving upon $\pi_k$ as in Lemma~\ref{lem:i-improve}. 
 If, on the other hand, all such $\pi_k$'s are fully informative, then whenever the receiver chooses a signal from any sender other than $i$,
 that signal is fully informative. But this implies that the receiver always has a fully-informative signal available, and, since such a signal is optimal for him,
 he will always choose it. This contradicts the assumption that $\ppi$ is not fully-informative.
 \end{proofof}

\subsection{Proof of Theorem~\ref{thm:mixed}}

\begin{proofof}{Theorem~\ref{thm:mixed}}
Fix a NE $\ppi$, and suppose towards a contradiction that $\ppi$ is not fully-informative.
%
%
%

Observe that all senders have the same minimal $\tau_i$, namely $\tau_i=\tau$. To see this, let $j$ be the sender for whom $\tau_j$
is maximal---namely, $\tau_j \geq \tau_k$ for all $k\in\{1,\ldots,n\}$---and suppose there is some sender $i$ for whom $\tau_i<\tau_j$.
But this implies that $\ppi_i$ has positive measure on those signals that lead to receiver payoffs in $[\tau_i,\tau_j)$, signals that are never chosen.
By Lemma~\ref{lem:not-chosen}, this is a contradiction.
%
%
%
%

Now, note that, since $\ppi$ is not fully-informative, 
the signals leading to receiver utility $\tau$ are not fully informative (since fully informative signals lead to utilities that are maximal for the receiver
and are thus always chosen).

Next, we will consider three cases, corresponding to whether all, some, or no sender strategies put positive measure on the set of signals that lead to receiver
utility exactly equal to $\tau$:
\begin{itemize}
\item[(i)] $P(v_R(\ppi_i)=\tau)=0$ for all $i\in \{1,\ldots,n\}$.
\item[(ii)] $P(v_R(\ppi_i)=\tau)>0$ and $P(v_R(\ppi_j)=\tau)=0$ for some $i,j\in \{1,\ldots,n\}$.
\item[(iii)]  $P(v_R(\ppi_i)=\tau)>0$ for all $i\in \{1,\ldots,n\}$.
\end{itemize}

We begin with case (i). Let $i$ be any sender, and note that the equality $P(v_R(\ppi_i)=\tau)=0$ implies that for all $\eps>0$, 
$$P\left(v_R(\ppi_i)\in \left[\tau,\tau+\eps\right]\right)>0$$
(for otherwise, $\tau$ would not be in the support of $v_R(\ppi_i)$).
Thus, sender $i$ chooses a signal $\pi_i$ with $v_R(\pi_i)\in \left[\tau,\tau+\eps\right]$ with positive probability. Furthermore,
the probability that $i$'s signal is chosen by the receiver, when he plays such a signal, approaches 0 with $\eps$:
$$\lim_{\eps\rightarrow 0} P\left(D(\ppi)=i \mid v_R(\pi_i)\in \left[\tau,\tau+\eps\right]\right) = 0.$$

This is because, as $\eps$ approaches 0, the probability that there is some $j\neq i$ and $\pi_j$ such that $v_R(\pi_j)>\tau+\eps$ approaches 1, and
the receiver will choose the sender offering the signal with the highest $v_R$. Thus, sender $i$'s strategy $\ppi_i$ has positive measure on a set of signals
 that is almost never chosen by the receiver. By Lemma~\ref{lem:not-chosen}, this is a contradiction.

We now consider case (ii), that $P(v_R(\ppi_i)=\tau)>0$ and $P(v_R(\ppi_j)=\tau)=0$ for some $i,j\in \{1,\ldots,n\}$.
Here, observe that signals $\pi_i\in\supp(\ppi_i \mid v_R(\ppi_i)=\tau)$ are played by sender $i$ with positive probability, but are chosen with probability
0 (since any signal of sender $j$ will be preferred over $\pi_i$).  As above, this again implies that sender $i$'s strategy $\ppi_i$ has positive measure on a set of signals
 that is almost never chosen by the receiver. Again by Lemma~\ref{lem:not-chosen}, this is a contradiction.

%

Finally, we consider case (iii), that  $P(v_R(\ppi_i)=\tau)>0$ for all $i\in \{1,\ldots,n\}$. 
Note that, under case (iii), $P\left(v_R\left(\pi_{D(\ppi)}\right)=\tau\right)>0$.


We consider three cases. First, if for some sender $i$ it holds that $E[v_i(\ppi_i)\mid v_R(\pi_{D(\ppi)})=\tau] \geq E[v_i(\ppi_j)\mid v_R(\pi_{D(\ppi)})=\tau]$
for all $j\neq i$, with a strict inequality for at least one such $j$,
then sender $i$ has the following profitable deviation:
Instead of choosing $(\ppi_i\mid v_R(\ppi_i)=\tau)$, sender $i$ chooses $\ppi_i'$: Replace every $\pi_i\in\supp(\ppi_i\mid v_R(\ppi_i)=\tau)$ with $\pi_i'$ such that with (arbitrarily) small probability $\eps$ signal $\pi_i'$ is realized as a fully-informative signal---that is, it is realized as $\arg\max_{a\in\A} u_R(a,\omega_R)$ in every state $\omega_R\in\Omega_R$---and with probability $1-\eps$ it is realized as $\pi_i$. 
	For any $\eps>0$, the new signal distribution $\ppi_i'$ will be more informative than $\ppi_j$, and so the receiver will choose sender $i$ with probability 1 (conditional on $v_R(\pi_{D(\ppi)})=\tau$), leading to a strict increase in $i$'s utility.
	Since $\eps$ can be arbitrarily small, this does not harm sender $i$ much, and so for small enough $\eps$ the overall gain from the deviation is positive.

The second case is if for some sender $i$ it holds that $E[v_i(\ppi_i)\mid v_R(\pi_{D(\ppi)})=\tau] < E[v_i(\ppi_j)\mid v_R(\pi_{D(\ppi)})=\tau]$ for
some $j\neq i$.
In this case, sender $i$ has the following profitable deviation: Replace $(\ppi_i\mid v_R(\ppi_i)=\tau)$ by $\ppi_i'$, where $\ppi_i'$ is the ``simulation''
of  $(\ppi_j\mid v_R(\ppi_j)=\tau)$. That is,  for each $\pi_j\in\supp(\ppi_j\mid v_R(\ppi_j)=\tau)$,
$i$ plays the simulation $\pi_i'$ of $\pi_j$ guaranteed to exist by Lemma~\ref{lem:simulate}. Under this deviation, sender $i$'s
utility, conditional on receiver utility equal to $\tau$, is equal to $E[v_i(\ppi_j)\mid v_R(\pi_{D(\ppi)})=\tau]$. By assumption, this is a
strict improvement over $E[v_i(\ppi_i)\mid v_R(\pi_{D(\ppi)})=\tau]$.

The third case is if for some sender $i$ it holds that
$E[v_i(\ppi_i)\mid v_R(\pi_{D(\ppi)})=\tau] = E[v_i(\ppi_j)\mid v_R(\pi_{D(\ppi)})=\tau]$
for all $j\neq i$. Note that this expected utility is thus equal to $E[v_i(\pi_{D(\ppi)})\mid v_R(\pi_{D(\ppi)})=\tau]$.
Let $\ell$ and $\pi_\ell$ satisfy $v_i(\pi_\ell)\geq v_i(\pi_k)$ for all senders $k\neq i$ and signals $\pi_k\in\supp(\ppi_k)$. That is, $\pi_\ell$
 is the signal yielding highest utility to sender $i$ played by any other sender.
In this third case sender $i$ has the following profitable deviation:  
If $\pi_\ell$ is not fully-informative, replace $(\ppi_i\mid v_R(\ppi_i)=\tau)$ by $\pi_i'$, where $\pi_i'$ is the improvement on the simulation of $\pi_\ell$ guaranteed to exist by Lemma~\ref{lem:i-improve}. 
By the lemma and assumption, $v_i(\pi_i')>v_i(\pi_\ell)\geq E[v_i(\ppi_i)\mid v_R(\pi_{D(\ppi)})=\tau]$.
Since sender $i$ deviates to a signal that yields him the highest possible utility over all other senders' signals, this deviation is profitable.

If $\pi_\ell$ is fully informative and $v_i(\pi_\ell)> E[v_i(\ppi_i)\mid v_R(\pi_{D(\ppi)})=\tau]$, replace $(\ppi_i\mid v_R(\ppi_i)=\tau)$ by a fully-informative signal.
This again is a profitable deviation. 

Finally, suppose $\pi_\ell$ is fully informative and $v_i(\pi_\ell)= E[v_i(\ppi_i)\mid v_R(\pi_{D(\ppi)})=\tau]=E[v_i(\pi_{D(\ppi)})\mid v_R(\pi_{D(\ppi)})=\tau]$.
 The equality $v_i(\pi_\ell)=E[v_i(\pi_{D(\ppi)})\mid v_R(\pi_{D(\ppi)})=\tau]$ implies that for all senders $k\neq i$ and signals $\pi_k\in\supp\left(\ppi_k \mid v_R(\ppi_k)=\tau\right)$ the equality
 $v_i(\pi_k)=E[v_i(\pi_{D(\ppi)})\mid v_R(\pi_{D(\ppi)})=\tau]$ holds. If some such $\pi_k$ is not fully-informative, then sender $i$ can invoke the same profitable deviation as
 above, by simulating and improving upon $\pi_k$ as in Lemma~\ref{lem:i-improve}. 
 If, on the other hand, all such $\pi_k$'s are fully informative, then whenever the receiver chooses a signal $\pi_k$ from any sender other than $i$, and obtains utility
 $v_R(\pi_k)=\tau$, that signal is fully informative. However, this is a contradiction, since $\ppi$ is not fully-informative and so 
 $\tau$ is strictly less than the receiver's optimal utility.
\end{proofof}

\section{Discussion and Summary}\label{sec:conc}

Underlying our results is an observation that competing senders always prefer to be chosen by the receiver in the interim stage. Consequently, they engage in a kind of a war of attrition and yield more and more information to the receiver until, finally, they disclose all of the receiver's payoff-relevant information.
This result holds whenever the senders are not perfectly aligned in their interests. The sender's advantage when he monopolizes the information becomes a knife-edge advantage with more than one sender. Whenever senders' interests are perfectly aligned they can still enjoy this advantage. However, any perturbation, no matter how small, will result in the complete transfer of the surplus to the receiver.

To model the possibility of senders' misalignment we consider a state space with a product structure, whereby each player's utility function only depends on one entry of the state. Note that, as we allow correlation between senders and receiver,  this model does not exclude the possibility that a sender and receiver care about the same set of states (as is standard in the single-sender model).
What our model allows is the introduction of conditions whereby one sender is not perfectly informed about the states of other senders, and so misalignment of interests is always possible.


\newpage
\bibliography{competing-senders-bib}
\appendix
\begin{Large}\begin{center}\textbf{Appendix}\end{center}\end{Large}

\section{Tie-Breaking Rule}\label{sec:decision-rule-extension}
In this section we discuss extensions of our main result to different tie-breaking rules of the receiver. 

We first observe that, although we assumed that when the receiver is indifferent between several senders' signals he chooses one uniformly at random,
this assumption is not necessary for our proof. In fact, for the proof to go through, we need the following property: For
every profile $\pi=(\pi_1,\ldots,\pi_n)$, sender $i$,  deviation $\pi_i'$, and profile $\pi'=(\pi_i',\pi_{-i})$, if $i\not\in\supp(D(\pi))$ and $i\not\in\supp(D(\pi'))$,
then the distributions $D(\pi)$ and $D(\pi')$ are identical. This property is satisfied by any tie-breaking rule that depends only on the identities and signals of the 
best senders, from the receiver's perspective. It could even involve different distributions for different top senders or different signals that they choose.

We note that without such a property, one problem that may arise is that a deviating sender may be ``punished"
by the receiver for deviating: For example, suppose $i$ is never chosen under profile $\pi$, nor under $\pi'$. However, the receiver may break ties differently
under $\pi'$ than under $\pi$ in such a way as to harm $i$ and so discourage him from deviating. 
This possibility is not accounted for in our proof of Theorem~\ref{thm:mixed}. We do not know whether our main result holds without this restriction
on tie-breaking rules.

For the special case in which we restrict to pure NE, however, we do have a stronger result---one that is even stronger than only allowing for
unrestricted tie-breaking rules. To describe this strengthening, denote by $\mathcal{D}$ the set of all optimal decision rules for the
receiver:
$$\mathcal{D} = \left\{D:\Pi_1\times\ldots\times\Pi_n\mapsto\Delta\left(\{1,\ldots,n\}\right):i\in\supp\left(D\left(\pi_1,\ldots,\pi_n\right)\right)\Rightarrow\pi_i\succeq_R \pi_j ~\forall j\in\{1,\ldots,n\} \right\}.$$
In words, given a profile of players' signals, an optimal decision rule $D$ of the receiver is any distribution over senders whose signals he prefers over others'.
We now define the notion of a (pure) pessimistic Nash equilibrium:

\begin{definition}
A {\em pure pessimistic Nash equilibrium} of the game is a profile  $\left(\pi_1,\ldots,\pi_n\right)$ and a receiver's decision rule $D\in\mathcal{D}$ 
with the following property: For every  sender $i$ and signal $\pi_i'$ there exists a decision rule $D'\in\mathcal{D}$ for which
$$E\left[u_i(a,\omega_i) \mid \pi_{D'(\pi_i',\pi_{-i})}\right]\leq E\left[u_i(a,\omega_i) \mid \pi_{D(\pi_1,\ldots,\pi_n)}\right],$$
where $(\pi_i',\pi_{-i})$ is the profile $\left(\pi_1,\ldots,\pi_n\right)$ but with $\pi_i'$ replacing $\pi_i$.
\end{definition}

Note the quantifier on $D'$: Counterfactually, a deviation is profitable if it is profitable {\em for every} decision rule following the deviation. Hence the qualifier in the definition above---a deviating sender is {\em pessimistic} about the receiver's reaction to his deviation. Note that the set of pure pessimistic Nash equilibria is a superset
of the set of pure NE, and so the weaker definition renders our main result stronger.

In particular, the proof of Theorem~\ref{thm:mixed} actually shows that all pure pessimistic NE of the game are fully informative. To see this, note that the proof,
when restricted to pure equilibria, does the following: In a non-fully-informative profile, whenever a sender is not chosen with probability 1, he can construct 
a profitable deviation under which he will be chosen with probability 1. That is, the deviation is such that his signal is strictly preferred by the receiver over all
other signals, and so the issue of tie-breaking is rendered moot.

%
%


\end{document}